\documentclass[11pt]{article}
\usepackage{fullpage}

\author{Kunal Talwar\\
Apple.\\
\texttt{ktalwar@apple.com}}

\usepackage{lmodern}
\usepackage{xspace}                 
\usepackage{ktmacros}
\usepackage{nicefrac}

\usepackage{floatpag}

\usepackage[italic,defaultmathsizes,symbolgreek]{mathastext}

\usepackage[boxruled,vlined,nofillcomment,linesnumbered]{algorithm2e}
\SetKwProg{Fn}{}{\string:}{}
\SetKwInput{KwParam}{Parameters}
\SetKwProg{KwProtocol}{Protocol}{\string:}{}

\usepackage[capitalise,nameinlink]{cleveref}

\newcommand{\Ell}{\ensuremath{\mathcal{L}}}
\newcommand{\edapprox}{\ensuremath{\approx_{(\eps,\delta)}}}
\newcommand{\edapproxp}[2]{\ensuremath{\approx_{(#1,#2)}}}

\newcommand{\vg}{\mathbf{g}}
\newcommand{\normal}{\mathcal{N}}
\newcommand{\vW}{\mathbf{W}}

\newcommand{\vs}{\mathbf{s}}

\DeclarePairedDelimiterX{\infdivx}[2]{(}{)}{#1\;\delimsize\|\;#2}

\crefname{lem}{Lemma}{Lemmas}
\crefname{algocf}{alg.}{algs.}
\Crefname{algocf}{Algorithm}{Algorithms}
\newif\iffull
\fullfalse

\begin{document}
\title{Differential Secrecy for Distributed Data and Applications to Robust Differentially Secure Vector Summation}
\date{}
\maketitle
\begin{abstract}

  Computing the noisy sum of real-valued vectors is an important primitive in differentially private learning and statistics. In private federated learning applications, these vectors are held by client devices, leading to a distributed summation problem. Standard Secure Multiparty Computation protocols for this problem are susceptible to poisoning attacks, where a client may have a large influence on the sum, without being detected.

  In this work, we propose a poisoning-robust private summation protocol in the multiple-server setting, recently studied in PRIO~\citep{Prio}. We present a protocol for vector summation that verifies that the Euclidean norm of each contribution is approximately bounded. We show that by relaxing the security constraint in SMC to a differential privacy like guarantee, one can improve over PRIO in terms of communication requirements as well as the client-side computation. Unlike SMC algorithms that inevitably cast integers to elements of a large finite field, our algorithms work over integers/reals, which may allow for additional efficiencies.
\end{abstract}
\section{Introduction}

We investigate the problem of distributed private summation of a set of real vectors, each of norm at most 1. Each client device holds one of these vectors and the goal is to allow a server to compute the sum of these vectors. Privacy constraints require that an adversary not learn too much about any of these vectors, and this constraint will be expressed as a differential privacy~\citep{DR14-book} requirement.

This is a common primitive to private federated learning and statistics. In a setting of a trusted server, the clients could send the vectors to the server, which could then output the sum with appropriate noise added to ensure differential privacy. A natural solution then is to use tools from secure multiparty computation to simulate this trusted server. This approach goes back to the early days of differential privacy~\citep{ODOpaper}, and has been heavily investigated~\citep{ChanSS12,Bonawitz17}. Practical protocols applying this approach have to deal with clients dropping out during the protocol, and often scale poorly with the number of clients. The security guarantee of SMC ensures that we learn nothing except the (noisy) sum. However, a malicious client in many of these protocols can contribute a vector with arbitrarily large norm and go completely undetected. Addressing this manipulability would require additional modification to these protocols, making them less feasible.

An elegant way out is possible under slightly stronger trust assumptions. \citet{Prio} show that if we have a set of $S$ servers where at least one of them is trusted, we can efficiently get both privacy and integrity\footnote{We defer the precise definitions to~\cref{sec:definitions}}. In our application, this framework gives a protocol that validates that each vector has norm at most 1, and computes the sum of vectors. The security guarantee here says that other than the output and the fact that the inputs have norm at most 1, any strict subset of servers learns nothing about the clients' inputs. If the clients add a small amount of noise to their inputs, or more generally, use a local randomizer, the final output can be shown to be differentially private. As long as all the inputs are bounded in norm, the validity predicates are all 1 and hence have no information. The overall security guarantee then says that the view of any strict subset of servers is differentially private with respect to the input vectors.

Note that perfect secrecy here is impossible as the output itself leaks information about the inputs. In the approach described above, {\em What we compute} does not leak too much about the input since we are computing a differentially private output. {\em How we compute it}, i.e. the computation protocol itself provides perfect secrecy {\em subject to} the output.

Our guarantee of interest is the leakage about any input from the process as well as the output, i.e. the sum of the privacy costs from the {\em what} and the {\em how}. In this work, we relax the secrecy guarantee of the protocol to a differential secrecy guarantee. We show that this allows for simpler and more efficient algorithms for the robust vector summation problem.

As a warm-up, we first show a natural variant of secret sharing that satisfies differential secrecy.
We next show that one can privately verify the norm of a secret-shared vector, if one allows some slack. We present a simple protocol based on random projections. Our protocol accepts all vectors of norm at most 1 with high probability. Additionally, a vector with too large a norm (polylogarithmic in the parameters) will be rejected with high probability. Thus we have some robustness: a malicious client can affect the sum by more than norm 1, but not arbitrarily more. Our privacy proof here relies on a new result on the privacy bounds for noisy random projections. Unusually for a differential privacy result, here we exploit the randomness of the ``query''. Compared to PRIO, our verification algorithm requires no additional work from clients, and requires less communication between servers.

With secret-sharing and norm-verification over secret shares in place, our algorithm for summation is simple. The clients secret-share their vectors, and the servers run the norm-verification protocol on all the clients. For the clients that pass the norm verification, each server adds up their secret shares. The servers now hold additive secret shares of the summation, which can be communicated between servers to derive the vector summation.

This then eliminates the need for the client to perform any additional computation ($\Theta(d)$ in PRIO) or communication ($\Theta(\sqrt{d})$~\cite{BonehBCGI19} in PRIO). The validity check comes at zero cost to the client. This comes at a small increase in the inter-server communication from $3$ field elements to a logarithmic number of real numbers.

Our algorithms can work over real numbers or integers, instead of finite fields. Compressing these to reduce communication, for example by truncating or rounding does not affect the privacy guarantee, allowing one to find a representation that provides an acceptable tradeoff between accuracy and communication cost.

In practice, as we discuss in~\cref{sec:robust_secagg}, this can be a significant saving, especially in settings such as federated learning where the vectors being aggregated are high-dimensional gradients and the client to server communication is often the bottleneck. For typical parameters, where PRIO would need a large finite field needing 128 bits per coordinate (or at the very least 32 bits per coordinate), using real numbers can bring us down to 8 or 16 bits per coordinate.

Several natural questions remain. Our norm verification, and hence our robustness guarantee for summation, is approximate. We reject vectors with large enough norm. It would be interesting to reduce, or even eliminate this approximation, while maintaining the efficiency advantages of our protocol. Given the practical relevance of robust summation, it would also be compelling to improve distributed proofs of norm bound in the standard PRIO setting.

Finally, relaxing perfect secrecy in secure multiparty computation, or more broadly in cryptography to differential secrecy may allow for more efficient protocols in other settings.

\section{Related Work}

The question of simultaneously studying the differentially private function (the {\em What}) and the cryptographic protocol for computing it (the {\em How}) was first studied by~\citet{BeimelNO08}. They showed that in the SFE setting without a trusted server, one can provably gain in efficiency of the protocol for summing $0$-$1$ values. This differential privacy-based definition of security was subsequently used by \citet{BackesKMR15}, who show that this relaxation allows one to use imperfect randomness in certain cryptographic protocols.

Private anonymous summation protocols using mutliple servers go back to at least the split-and-mix protocol of~\citet{IshaiKRS06}. In the context of differential privacy, these have gained a lot of importance given recent results in the shuffle model of privacy~\citep{Bittau17, ErlingssonFMRTT19,CheuSUZZ19,BalleBGN19a}. Recent works by~\citet{Balle2020, GMPV20, GhaziKMPS21} have improved the efficiency of these results. These protocols however suffer from the manipulability issue: it is easy for one malicious client to significantly poison the sum without getting detected.

Another line of work~\citep{Bonawitz17} proposes practical secure summation protocol under different trust assumptions. These protocols also suffer from the manipulability problem. Recent works such as~\citep{SoGA20,BellBGLR20} address the scaling challenges in that work.

The two-party version of some of these questions have been studied by~\citep{McGregor, GoyalMPS}.~\citet{KairouzOVa, KairouzOVb} study private secure multiparty computation under a local differential privacy constraint. In a different vein,~\citet{CheuSU19} show that locally differentially private algorithms are fairly manipulable by small subsets of users, and quantify their manipulability.

\section{Definitions}
\label{sec:definitions}
We would like the protocol to satisfy several properties. We define appropriate notions of these first.

\begin{definition}[Completeness]
  A protocol $\Pi$ is $(1-\beta)$-complete w.r.t. $\Ell$ if for all $x \in \Ell$, the protocol accepts $x$ with probability at least $(1-\beta)$.
\end{definition}

\begin{definition}[Soundness]
  A protocol $\Pi$ is $\beta$-sound w.r.t. $\Ell$ if for $x \not\in \Ell$, the protocol accepts with probability at most $\beta$.
\end{definition}

Let $\Ell_{r}$ denote the set of vectors with norm at most $r$. We will show completeness w.r.t. $\Ell_1$ and soundness w.r.t. $\Ell_\rho$. for a parameter $\rho >1$.

Additionally, we would like a mild relaxation of Zero Knowledge, inspired and motivated by the notion of Differential Privacy. We first recall a notion of near-indistinguishability used in Differential Privacy:

\begin{definition} Two random variables $P$ and $Q$ are said to be $(\eps,\delta)$-close, denoted by $P \approx_{(\eps, \delta)} Q$ if for all events $S$, $Pr[P \in S] \leq \exp(\eps) \cdot \Pr[Q \in S] + \delta$, and similarly, $Pr[Q \in S] \leq \exp(\eps) \cdot \Pr[P \in S] + \delta$
\end{definition}

One can relax the secrecy requirements in cryptography to differential secrecy.  Here we define this notion for Zero Knowledge\footnote{This is the {\em local DP} version of ZK which is appropriate in this setting. One can similarly define a central DP version, where the simulator has access to all but one client's input}.
\begin{definition}
  We say a protocol $\Pi$ is $(\eps,\delta)$-Differentially Zero Knowledge w.r.t. $\Ell$ if there is a distribution $Q$ such that for all $x \in \Ell$, the distribution $\Pi(x)$ of the protocol's transcript on input $x$ satisfies $\Pi(x) \edapprox Q$.
\end{definition}

Note that here we require privacy, or differential zero knowledge for $x \in \Ell$. While one can naturally define a computational version of this definition, along the lines of computational differential privacy definitions~\cite{MironovPRV09}, we restrict ourselves to the information-theoretic version in this work.

In this work, we will be using multi-verifier protocols. Here the notion of near Zero Knowledge is with respect to a strict subset of verifiers.
\begin{definition}
A single-prover, multiple-verifier protocol $\Pi$ is $(\eps, \delta)$-Differentially Zero Knowledge w.r.t to a subset $T$ of parties if there is a distribution $Q$ dependent only on inputs of $T$ and the output of the protocol, such that for any set of inputs for $T^c$ that are valid for some $x\in \Ell$, the distribution of messages from $T^c$ to $T$ is $(\eps, \delta)$-close to $Q$.
\end{definition}

\medskip\noindent{\bf Attack Models}: In our work, the client will play the role of the prover, and the servers will play the role of the verifiers. We interchangeably use client/server and prover/verifier terminology as appropriate. We will prove completeness and privacy for honest-but-curious prover. We will establish soundness against an arbitrary malicious provers. This implies that a client that is behaving according to the protocol will get a strong privacy guarantee, and will be accepted with high probability. A malicious client will still likely be caught, and may not get a privacy assurance. Our protocols will have privacy against an a strict subset of servers being malicious, as long as at least one of the servers is honest. The soundness and completeness results will assume that all servers are honest. Thus some subsets of servers behaving maliciously can hurt the utility of the protocol, but not the privacy.

We remark that there is a definitional choice here: when defining a zero-knowledge protocol with soundness/completeness strictly smaller than 1, the simulator may be given access to $\one(x \in \mathcal{L})$, or to the output of the protocol. This leads to slightly different definitions of zero knowledge.
Since we want privacy only for $x \in \mathcal{L}$, the first version would essentially mean that $T$ can differentially simulate the full interaction. The second definition allows leakage of the output of the protocol. While we are typically interested in the former for the whole protocol, in this work we choose the second option. This modular approach allows us to separately analyze the privacy cost of the output of the protocol. In particular, we may apply different analyses depending on whether we consider distributed noise addition, or apply local randomizers and rely on privacy amplification by shuffling. We defer additional discussion to~\cref{sec:robust_secagg}.

\subsection{Secure Summation}
The secure summation problem is defined as follows. There is a set of $N$ clients with client $i$ holding a vector $\vx_i \in \Re^d$ with $\|\vx_i\| \leq 1$. Our goal is to design a protocol with $S$ servers such that for suitable parameters $\eps, \delta, \rho, \beta$, the following properties hold:
\begin{description}
  \item [Correctness:] When all parties are honest, the protocol allows a designated server to compute a vector $\vy \in \Re^d$ such that  $\vy = \sum_i \vx_i$ with probability at least $(1-\beta)$.
  \item [Privacy:] For any honest client $i$, the protocol is $(\eps, \delta)$-Differentially Zero Knowledge w.r.t. any subset of parties that excludes at least one server.
  \item [Robustness:] For any possibly malicious client $i$, the computed summation $\vy$ differs from the output $\vy_{-i}$ without client $i$ in norm by at most $\rho$, i.e. $\|\vy-\vy_{-i}\|_2 \leq \rho$, except with probability at most $\beta$.
\end{description}

In words, we would like a protocol that is private w.r.t. to any honest client as long as at least one of the $S$ servers is honest. Thus an honest client that trusts at least one of the servers to be honest is assured of a differential privacy guarantee. The robustness property gives an integrity guarantee if all servers are honest. The parameter $\rho \geq 1$ controls how much any client can impact the output of the protocol. Note that a malicious client can always behave as if their input was $\vx'_i$ for any arbitrary vector of norm $1$. The robustness requirement here puts an upper bound on how much a malicious client can distort the summation The correctness and robustness properties will allow failure with probability $\beta$. Depending on the application, a small constant $\beta$ may be acceptable.

\section{Preliminaries}
We state two important properties of the differential privacy notion of closeness.
\begin{proposition}
  Suppose that $P \edapprox Q$ and $P' \edapproxp{\eps'}{\delta'} Q'$. Then
  \begin{enumerate}
    \item[Post Processing:] For any function $f$, $f(P) \edapprox f(Q)$.
    \item[Simple Composition: ] $(P, P') \edapproxp{\eps+\eps'}{\delta+\delta'} (Q, Q')$.
  \end{enumerate}
\end{proposition}
\noindent The following is a restatement of the privacy of the Gaussian mechanism~\citep[Thm A.1]{DR14-book}.
\begin{lemma}
\label{lem:gaussian}
Let $\eps, \delta >0$ and let $\vx \in \Re^d$ satisfy $\|\vx\|_2 \leq 1$. Let $P \sim \normal(\mathbf{0}, \sigma^2 \mathbb{I}_d)$ and let $Q \sim \vx + \normal(\mathbf{0}, \sigma^2 \mathbb{I}_d)$.
Then $P \edapprox Q$ if $\sigma \geq 2\sqrt{\ln \frac 2 \delta}/{\eps}$.
\end{lemma}
\noindent We next prove the following simple result on the privacy properties of noisy random projections.
\begin{lemma}
  \label{lem:proj_gaussian}
Let $G$ be a random matrix in $\Re^{k\times d}$ such that for a constant $c_{\delta}$, every $\vx \in \Re^d, \|\vx\|\leq 1$ satisfies
  \begin{align}
    \Pr[\|G\vx\| \geq c_\delta] \leq \delta,\label{eq:subg_tails}
  \end{align}
  where the probability is taken over the distribution of $G$.
  Let $\sigma = 2 c_\delta \sqrt{\ln \frac 2 \delta}/{\eps}$. Then for any $\vx \in \Re^d$ with $\vx \leq 1$,
  \begin{align*}
    (G, \normal(\mathbf{0}, \sigma^2 \mathbb{I}_d)) \edapproxp{\eps}{2\delta} (G, G\vx + \normal(\mathbf{0}, \sigma^2 \mathbb{I}_d)).
  \end{align*}
\end{lemma}
\begin{proof}
  Fix $\vx$ and let $\mathcal{E}$ be the event that $\|G\vx\| \geq c_\delta$. By~\cref{lem:gaussian}, we have that conditioned on the event $\mathcal{E}$,
  \begin{align*}
    (G, \normal(\mathbf{0}, \sigma^2 \mathbb{I}_d)) \edapproxp{\eps}{\delta} (G, G\vx + \normal(\mathbf{0}, \sigma^2 \mathbb{I}_d)).
  \end{align*}
  By~\cref{eq:subg_tails}, $\Pr[\mathcal{E}]\geq 1-\delta$. The claim now follow from the definition of $(\eps,\delta)$-closeness.
\end{proof}
\noindent We next recall a version of the Johnson-Lindenstrauss lemma on the length of random projections.
\begin{lemma}[Gaussian Ensemble JL]
  \label{lem:gaussian_ensemble}
  Let $G \in \Re^{k \times d}$ be a random matrix where each $G_{ij} \sim \normal(0, \frac{1}{k})$. Then for any $\vx \in \Re^d$ with $\|\vx\|\leq 1$,
    \begin{align*}
      \Pr[\|G\vx\| \not\in (1\pm O(\sqrt{{(\ln \tfrac 1 \delta)} /{k}}))\|\vx\| ] \leq \delta
    \end{align*}
\end{lemma}
To get more precise estimates, we recall that the sum of squares of $k$ $\normal(0, \frac 1 k)$ random variables is distributed as a (scaled version of a) chi-square distribution $\chi_k^2$. We will use the following tail bounds for $\chi_k^2$ random variables from \citet[Lemma 1 rephrased]{LaurentM2000}:
\begin{theorem}
  \label{thm:chi2_tails}
  Let $Q$ be a $\chi_k^2$ random variable. Then for any $\beta >0$,
  \begin{align*}
    \Pr[\frac{1}{k}Q \leq 1 - 2\sqrt{x/k}] &\leq \exp(-x),\\
    \Pr[\frac{1}{k}Q \geq 1 + 2\sqrt{x/k} + 2x/k] &\leq \exp(-x).
  \end{align*}
  \end{theorem}
\noindent Combining~\cref{thm:chi2_tails} with~\cref{lem:proj_gaussian}, we get the following useful corollary.
\begin{corollary}
  \label{cor:guassian_jl}
  Let $G \in \Re^{k \times d}$ be a random matrix where each $G_{ij} \sim \normal(0, \frac{1}{k})$ and let $c_\delta = \sqrt{1+2\sqrt{(\ln \tfrac 1 \delta)/{k}} + 2(\ln \tfrac 1 \delta)/{k}}$. Let $\sigma = 2 c_\delta \sqrt{\ln \frac 2 \delta}/{\eps}$. Then for any $\vx \in \Re^d$ with $\vx \leq 1$,
  \begin{align*}
    (G, \normal(\mathbf{0}, \sigma^2 \mathbb{I}_d)) \edapproxp{\eps}{2\delta} (G, G\vx + \normal(\mathbf{0}, \sigma^2 \mathbb{I}_d)).
  \end{align*}
\end{corollary}

\section{Warm-up: Secret Sharing Real-valued Vectors}
As a prelude to our result on norm verification, we first show how the standard secret sharing protocol extends to real-valued vectors, when allowing for Differential secrecy. Consider the protocol for secret-sharing a real-valued vector of norm at most $1$ between $S$ servers shown in Algorithm~\ref{alg:secret-share}.
\begin{algorithm}\DontPrintSemicolon
  \caption{Secret Sharing a real vector.}
  \label{alg:secret-share}

  \BlankLine
  \Fn{Prover($\vx$)}{
  \KwIn{Vector $\vx \in \Re^d$ with $\|\vx\| \leq 1$.}
  \KwParam{$\sigma_{SS} \in \Re$.}
  Generate $\vg_1, \ldots, \vg_{S-1} \sim \normal(\mathbf{0}, \sigma_{SS}^2 \mathbb{I}_d)$ using private randomness.\\
  Send $\vx - \sum_{i=1}^{S-1}\vg_i$ to Verifier 0.\\
  \For{$i=1\ldots S-1$}{
  Send $\vg_{i}$ to Verifier $i$.
  }
  }
\end{algorithm}

To prove the differential secrecy for this protocol, we show a simulator for any subset of verifiers in Algorithm~\ref{alg:secret-share-sim}.
\thispagestyle{empty}
\begin{algorithm}\DontPrintSemicolon
  \caption{Simulator for Algorithm~\ref{alg:secret-share}.}
  \label{alg:secret-share-sim}

  \BlankLine
  \Fn{Simulator($T \subsetneq [S]$)}{
  \KwIn{$T$ proper subset of $S$}
  \KwParam{$\sigma_{SS} \in \Re$.}
  \For{$i \in T$}{
  \If{$i \neq 0$}{
     Generate $\vg_i \sim \normal(\mathbf{0}, \sigma_{SS}^2 \mathbb{I}_d)$.\\
     Send $\vg_i$ to Verifier $i$.
  }
  }
  \If{$0 \in T$}{
  Generate $\vg \sim \normal(\mathbf{0}, (S-|T|) \sigma_{SS}^2\mathbb{I}_d)$.\\
  Send $\vg  - \sum_{i\in T; i \neq 1} \vg_i$ to Verifier 0.\\
  }
  }

\end{algorithm}
We next argue that this secret sharing scheme is differentially secure.
\begin{theorem}
  \label{thm:dzk_ss}
  Fix any $T \subsetneq [S]$. Then Prover($\vx$)$|_T \edapprox$ Simulator($T$) for $(S-|T|)\sigma_{SS}^2 \geq 4 \ln \frac 2 \delta / \eps^2$.
\end{theorem}
\begin{proof}
  If $0 \not\in T$, the simulation is perfect: indeed each verifier in $T$ receives an independent Gaussian vector with variance $\sigma_{SS}^2\mathbb{I}_d$ in both distributions.  When $0 \in T$, consider the distribution of the message to Verifier $0$ conditioned on $T \setminus [0]$.

  The simulator output to Verifier $0$ is distributed as $\normal(-\sum_{i \in T; i\neq -0} \vg_i, (S-|T|) \sigma_{SS}^2\mathbb{I}_d)$. The message to verifier 0 from the prover, conditioned on
  $\{\vg_{i}\}_{i\in T: i \neq 0}$ is distributed as $\normal(\vx-\sum_{i \in T; i\neq 0} \vg_i, (S-|T|) \sigma_{SS}^2\mathbb{I}_d)$. The claim now follows from the privacy of the Gaussian mechanism (\cref{lem:gaussian}).
\end{proof}
The differential secrecy implies that an honest prover's privacy is protected against an arbitrary collusion of verifiers short of all of them. Note also that by making $\sigma_{SS}$ larger, we can improve the privacy cost. A larger $\sigma_{SS}$ only costs us in terms of the precision to which these messages should be communicated to ensure that the sum of secret shares is close to $\vx$. Note that we can post-process these vectors (both in the algorithm and its simulation), e.g. by rounding or truncation. By the post-processing property of differential privacy, the differential secrecy is maintained.

\section{Differential Zero Knowledge Proofs of bounded norm}

We next describe our DZK protocol to verify a Euclidean norm bound.
The first step is to secret-share the vector between the two verifiers as in the previous section. The rest of the protocol only involves the verifiers; the prover code therefore is identical to secret-sharing.

The second step is norm estimation and happens amongst the verifiers.
As a first cut, suppose that the servers aggregate their shares, while adding noise to each share to preserve privacy. This would require adding $d$-dimensional gaussian noise to each share. This noise being fresh and independent will contribute to the norm of the computed sum, which will now be about $\sqrt{d}$, and will have variance growing polynomially with $d$. This will make it impossible to estimate the norm better than some polynomial in $d$, and thus our gap $\rho$ will grow polynomially with the dimension.

To improve on this, we will use random projection into a $k$-dimensional space for a parameter $k$ independent of the dimension. Being a lower-dimensional object, a projection can be privately estimated much more accurately. The choice of the projection dimension $k$ will give us a trade-off between the privacy parameters and the gap assumption. Intuitively, we rely on the Johnson-Lindenstrauss lemma, which says that the Euclidean norm of a vector is approximately preserved under random projections. Since projection is a linear operator, computing the projection of a secret-shared vector is straight-forward. Verifier $0$ here takes the special role of collecting an estimate of a random projection of $\vx$, computing its norm and sharing the {\textsf Accept/Reject} bit.

\begin{algorithm}\DontPrintSemicolon
  \caption{Protocol for Norm Verification}
  \label{alg:norm-verify}
\thisfloatpagestyle{empty}
  \SetKwData{Accept}{Accept}
  \KwIn{Prover has a vector $\vx \in \Re^d$}
  \KwOut{Verifiers must agree on \Accept.}

  \BlankLine
  \Fn{Prover($\vx$)}{
  \KwIn{Vector $\vx \in \Re^d$ with $\|\vx\| \leq 1$.}
  \KwParam{$\sigma_{SS} \in \Re$.}
  Generate $\vg_1, \ldots, \vg_{S-1} \sim \normal(\mathbf{0}, \sigma_{SS}^2 \mathbb{I}_d)$ using private randomness.\\
  Send $\vx - \sum_{i=1}^{S-1}\vg_i$ to Verifier 0.\\
  \For{$i=1\ldots S-1$}{
  Send $\vg_{i}$ to Verifier $i$.
  }
  }
\setcounter{AlgoLine}{0}
\BlankLine
\hrule
\BlankLine
  \Fn{Verifier-0}{
  \KwParam{Integer $k$. Threshold $\tau \in \Re$.}
  Receive $\vz_0$ from Prover. \tcp*{Expected to be $\vx - \sum_{i=1}^{S-1}\vg_i$}
  Generate $\vW \in \Re^{k\times d}$ with each $W_{ij} \sim \normal(0, \frac{1}{k})$ using private randomness.\\
  \tcp*{This version assumes honest Verifier 0. To allow malicious Verifier 0, $W$ is generated using randomness shared amongst verifiers.}
  Send $\vW$ to Verifiers $1, \ldots, S-1$.\\
  \For{$i=1\ldots S-1$}{
  Receive $\vy_i$ from Verifier $i$. \tcp*{Expect $\vy_i = \vW  \vg_i + Noise$.}
  }
  Compute $\vv = \vW \vz_0 + \sum_{i=1}^{S-1} \vy_i + \normal(0, \sigma_v^2 \mathbb{I}_k)$. \tcp*{Expect $\vv = \vW\vx + Noise$.}
  \If{$|\vv| \geq \tau$}{\Accept = 0}\Else{\Accept = 1}
  Send \Accept to Verifiers $1,\ldots, S-1$.
}
\setcounter{AlgoLine}{0}
\vspace{4pt}
\hrule
\vspace{4pt}
  \Fn{Verifier-$i$ ($i \geq 1$)}{
  \KwParam{Integer $k$. Noise scale $\sigma_{v} \in \Re$.}
  Receive $\vz_i$ from Prover. \tcp*{Expected to be $\vg_i$}
  Receive $\vW \in \Re^{k\times d}$ from Verifier-0.\\
  Compute $\vy_i = \vW \vz_i + \normal(0, \sigma_v^2 \mathbb{I}_k)$.\\
  Send $\vy_i$ to Verifier-0.\\
  Receive \Accept from Verifier-0.\\
}

\end{algorithm}

\newcommand{\Accept}{\normalfont{\mathrm{\textsf{Accept}}}}
We start with establishing compeleteness, which will determine the acceptance threshold $\tau$. We will then show soundness for an appropriate $\rho$.
\begin{theorem}[Completeness]
  \label{thm:completeness}
  Suppose that the prover and the verifiers are honest and the $\|\vx\|\leq 1$. Then for $\tau \geq \sqrt{(\frac{1}{k}+|S|\sigma_v^2)(k+2\ln\frac 1 \beta + 2\sqrt{k\ln\frac 1 \beta})}$,
  \begin{align*}
    \Pr[\Accept = 1] \geq 1 - \beta.
  \end{align*}
\end{theorem}
\begin{proof}
Under the assumptions, $W\vx$ is distributed as $\normal(0, \frac{\|\vx\|_2^2}{k} \Id)$. The noise added by each server is distributed as $\normal(0, \sigma_v^2\Id)$, and all of these Gaussian random variables are independent. Thus $\vv$ computed by Verifier 0 is distributed as $\normal(0, (\frac{\|\vx\|_2^2}{k} + |S|\sigma_v^2)\Id)$, and its squared norm is distributed as $(\frac{\|\vx\|_2^2}{k} + |S|\sigma_v^2) Q$, where $Q$ is a $\chi_k^2$ random variable. Thus
\begin{align*}
  \Pr[\|\vv\|_2^2 \geq \tau^2] &= \Pr[Q \geq (\frac{1}{k}+|S|\sigma_v^2)^{-1}\tau^2].
\end{align*}
Plugging the upper tail bounds from~\cref{thm:chi2_tails}, the result follows.
\end{proof}

\begin{theorem}[Soundness]
  \label{thm:soundness}
  Suppose that the verifiers are honest and suppose that $\|\sum_{i=0}^{S-1} z_i\| \geq \rho$, where $z_i$ is the message to verifier $i$. Then for $\rho^2 \geq \frac{k\tau^2}{k-2\sqrt{k\ln \frac 1 \beta}} - k|S|\sigma_v^2$,
  \begin{align*}
    \Pr[\Accept =  1] \leq \beta.
  \end{align*}
\end{theorem}
\begin{proof}
  As in the proof of~\cref{thm:completeness}, now $\|\vv\|_2^2$ is distributed as $(\frac{\rho^2}{k} + |S|\sigma_v^2)Q$ for a $\chi_k^2$ random variable $Q$. Using the lower tail bounds from~\cref{thm:chi2_tails}, it suffices to ensure
  \begin{align*}
    (\frac{\rho^2}{k} + |S|\sigma_v^2)(k-2\sqrt{k\ln \frac 1 \beta}) &\geq \tau^2.
  \end{align*}
  Rearranging, the claim follows.
\end{proof}

Some discussion on $k$ is in order. A small $k$ ensures that we need to add less noise and thus get better estimates. At the same time, larger $k$ ensures stronger concentration of the $\chi_k^2$ random variable. For intuition, we next estimate the bound on $\rho^2$ from~\cref{thm:soundness}, plugging in $\tau$ from~\cref{thm:completeness}. Setting $\lambda = \frac{\sqrt{\ln \frac 1 \beta}}{k}$ and assuming $\lambda$ is small enough, we can write
\begin{align*}
  \rho^2 &= \frac{k\tau^2}{k-2\sqrt{k\ln \frac 1 \beta}} - k|S|\sigma_v^2\\
  &= (1+k|S|\sigma_v^2)\frac{k+2\ln\frac 1 \beta + 2\sqrt{k\ln\frac 1 \beta}}{k-2\sqrt{k\ln \frac 1 \beta}} - k|S|\sigma_v^2\\
  &= (1+k|S|\sigma_v^2)\frac{1+2\lambda + 2\sqrt{\lambda}}{1-2\sqrt{\lambda}} - k|S|\sigma_v^2\\
  &\approx (1+k|S|\sigma_v^2)(1+O(\sqrt{\lambda})) - k|S|\sigma_v^2\\
  &= 1+O(k|S|\sigma_v^2\sqrt{\lambda})\\
  &\approx 1 + O(|S|\sigma_v^2 k^{\frac 1 2} (\ln \frac 1 \beta)^{\frac 1 4})
\end{align*}
Taking $k=\Theta(\sqrt{\ln \frac 1 \beta})$ suffices to ensure $\lambda$ is small enough for the approximations above to be valid. This leads to $\rho^2 = \Theta(|S|\sigma_v^2 \sqrt{\ln \frac 1 \beta})$. In practice, one may want to use the exact cdf for the $\chi_k^2$ distribution instead of the tail bounds used in the theorems.

\begin{algorithm}\DontPrintSemicolon
  \caption{Simulator for Algorithm~\ref{alg:norm-verify}.}
\label{alg:norm-verify-sim}

  \BlankLine
  \Fn{Simulator($T \subsetneq [S]; 0 \not\in T$)}{
  \KwIn{$T$ proper subset of $S$}
  \KwParam{$\sigma_{SS}, \sigma_v, \tau \in \Re$, integer $k$.}
  \SetKwData{Accept}{Accept}

  \For{$i \in T$}{
     Generate $\vg_i \sim \normal(\mathbf{0}, \sigma_{SS}^2 \mathbb{I}_d)$.\\
     Send $\vg_i$ to Verifier $i$.
  }
  Generate $\vW \in \Re^{k\times d}$ with each $W_{ij} \sim \normal(0, \frac{1}{k})$.\\
  Send $\vW$ to each Verifier in $T$.\\
  Receive $\{\vy_i\}_{i \in T}$.\\
  Compute $\vv_{Sim} = \sum_{i \in T} (\vy_i - \vW \vg_i) + \normal(\mathbf{0}, (S-|T|)\sigma_v^2 \mathbb{I}_k)$.\\
  \If{$|\vv_{Sim}| \geq \tau$}{\Accept = 0}\Else{\Accept = 1}
  Send \Accept to all verifiers.
}

\end{algorithm}
We now prove the differential zero knowledge property of the algorithm. We assume that verifier $0$ is honest. We will then relax this assumption using shared randomness.
\begin{theorem}[DZK assuming honest Verifier $0$]
  \label{thm:dzk_not_zero}
  Suppose that $\|\vx\|_2 \leq 1$. If the prover and Verifier-0 are honest, then for any $T \subset [S]\setminus \{0\}$, $T$'s view is $(\eps,\delta)$-DZK as long as $\sigma_v \geq 2 c_\delta \sqrt{\ln \frac 4 \delta}/{\eps}$.
\end{theorem}
\begin{proof}
  The simulator is defined in Algorithm~\ref{alg:norm-verify-sim}. The simulator sends messages to verifiers in $T$ in steps 4, 6, and 13. The messages in steps 4 and 6 follows exactly the same distribution as that in the mechanism, with all $\vg_i$'s and the matrix $\vW$ being independent normal. The message in step 13 is the $\Accept$ bit, which is computed as a post-processing of the vector $\vv_{Sim}$ computed in step 8. The corresponding $\Accept$ bit in the protocol is obtained by the same post-processing of $\vv$ computed by Verifier 0 in step $6$. Since the prover is honest, we can write:
  \begin{align*}
    (\vW, \vv) &= (\vW, \vW \vz_0 + \sum_{i=1}^{S-1} \vy_i + \normal(0, \sigma_v^2 \mathbb{I}_k))\\
     &= (\vW, \vW (\vx - \sum_{i=1}^{S-1} \vg_i) + \sum_{i=1}^{S-1} \vy_i + \normal(0, \sigma_v^2 \mathbb{I}_k))\\
      &= (\vW, \vW (\vx - \sum_{i=1}^{S-1} \vg_i) + \sum_{i\in T; i \neq 0} \vy_i + \sum_{i\not\in T; i \neq 0} \vy_i + \normal(0, \sigma_v^2 \mathbb{I}_k))\\
       &= (\vW, \vW\vx  + \sum_{i\in T; i \neq 0} (\vy_i - \vW \vg_i) + \sum_{i\not\in T; i \neq 0} (\vy_i  - \vW \vg_i)+ \normal(0, \sigma_v^2 \mathbb{I}_k))\\
       &= (\vW, \vW\vx + \sum_{i\in T; i \neq 0} (\vy_i - \vW \vg_i) + \sum_{i\not\in T; i \neq 0} \normal(0, \sigma_v^2 \mathbb{I}_k) + \normal(0, \sigma_v^2 \mathbb{I}_k))\\
       &= (\vW, \vW\vx + \sum_{i\in T; i \neq 0} (\vy_i - \vW \vg_i) +  \normal(0, (S - |T|)\sigma_v^2 \mathbb{I}_k))\\
       &= (\vW, (\vW\vx  +  \normal(0, (S - |T|)\sigma_v^2 \mathbb{I}_k))  + \sum_{i\in T; i \neq 0} (\vy_i - \vW \vg_i))\\
       &\edapproxp{\eps}{\delta} (\vW, \normal(0, (S - |T|)\sigma_v^2 \mathbb{I}_k)  + \sum_{i\in T; i \neq 0} (\vy_i - \vW \vg_i))\\
       &= (\vW, \vv_{Sim})
  \end{align*}
  Here we have used Corollary~\ref{cor:guassian_jl} in the second to last step.
\end{proof}
The honest prover assumption is necessary to give privacy to the prover. The assumption on Verifier 0 being honest is necessary as well in the protocol as stated: a malicious Verifier 0 that can choose an adversarial $\vW$ can violate the privacy constraint. For example, a verifier that knows that the true $\vx$ lies in a certain $k$-dimensional subspace can choose the projection matrix $\vW$ to project to that subspace. This will make the projected vector to have length much larger than $1$, and invalidate the assumptions in~\cref{lem:proj_gaussian}. We next show that this is the only place where we need Verifier 0 to be honest. Thus given a distributed oracle for randomly selecting $\vW$, e.g. using shared randomness, we have privacy as long as one of the Verifiers is honest.

\begin{theorem}[DZK assuming randomly chosen $\vW$]
  \label{thm:dzk_zero}
  Suppose that $\|\vx\|_2 \leq 1$. Further suppose that the prover is honest and the matrix $\vW$ shared in Step 4 by Verifier $0$ is uniformly random. Then for any $T \subsetneq [S]$, $T$'s view is $(\eps+\eps',\delta+\delta')$-DZK as long as $\sigma_v \geq 2 c_\delta \sqrt{\ln \frac 4 \delta}/{\eps}$ and $\sigma_{SS} \geq 2 \sqrt{\ln \frac 2 {\delta'}} / \eps'$
\end{theorem}\begin{proof}
The proof is nearly identical to the previous proof. When $0 \not\in T$, the theorem follows from~\cref{thm:dzk_not_zero}. When Verifier 0 is in the set, the secret sharing itself is $(\eps',\delta')$-DZK, by~\cref{thm:dzk_ss}. The rest of the protocol is $(\eps,\delta)$-DZK by repeating the proof of~\cref{thm:dzk_not_zero}. The result follows.
\end{proof}

\section{Application to Robust Secure Aggregation}
\label{sec:robust_secagg}
Our protocol for robust secure aggregation (Algorithm~\ref{alg:sec_agg}) builds on the additive secret shares with norm bound verification. The prover part of the protocol is nearly identical to secret sharing, with the only change being that the client sends its identifier $j$ with all the shares. We assume that each client has a unique identifier, though this assumption can be easily relaxed by having the client send a random nonce instead of its identifier.

The verifiers execute the norm verification protocol for each client. Verifier 0 constructs the set of indices $J^*$ that pass the norm verification and shares it with all the verifiers. The verifiers optionally check that $J^*$ is large enough; this part is not needed for our summation protocol, but can be useful to ensuring that the sum itself is differentially private. The verifiers now add up the secret shares for the provers in $J^*$ and share the sum with Verifier 0, that adds up the sums of secret shares to derive the sum.

\begin{algorithm}\DontPrintSemicolon
  \thisfloatpagestyle{empty}
  \caption{Client Protocol for Robust Secure Aggregation}
\label{alg:sec_agg}
  \SetKwData{Accepti}{Accept$_j$}
  \KwIn{Prover $j$ has a vector $\vx_j \in \Re^d$}
  \KwOut{Verifiers compute $\sum_j \vx_j$}

  \BlankLine
  \Fn{Prover$_j$($\vx_j$)}{
  \KwIn{Vector $\vx_j \in \Re^d$ with $\|\vx_j\| \leq 1$.}
  \KwParam{$\sigma_{SS} \in \Re$.}
  Generate $\vg_1, \ldots, \vg_{S-1} \sim \normal(\mathbf{0}, \sigma_{SS}^2 \mathbb{I}_d)$ using private randomness.\\
  Send $\vx_j- \sum_{i=1}^{S-1}\vg_i$ to Verifier 0.\\
  \For{$i=1\ldots S-1$}{
  Send $(j, \vg_{i})$ to Verifier $i$.
  }
  }
\end{algorithm}
\begin{algorithm}
  \caption{Server Protocol for Robust Secure Aggregation}

  \SetKwData{Accepti}{Accept$_j$}
  \KwIn{Prover $j$ has a vector $\vx_j \in \Re^d$}
  \KwOut{Verifiers compute $\sum_j \vx_j$}
\setcounter{AlgoLine}{0}
\BlankLine
  \Fn{Verifier-0}{
  \KwParam{Integer $k$. Threshold $\tau \in \Re$.}
  Receive $V_0 = \{(j, \vz_0^j)\}$ from Provers. Let $J_0 = \{j : (j, \vz_0^j) \in V_0\}$.\\
  Generate $\vW \in \Re^{k\times d}$ with each $W_{ij} \sim \normal(0, \frac{1}{k})$ using private randomness.\\
  Send $\vW$ to Verifiers $1, \ldots, S-1$.\\
  \For{$i=1,\ldots, S-1$} {
  Receive $V_i = \{(j, \vy_i^j)\}$ from Verifier $i$. Let $J_i = \{j : (j, \vy_i^j) \in V_i\}$.
  }
  Let $J = \cap_{i} J_i$.\\
  \For{$j \in J$}{
  Compute $\vv^j = \vW \vz_0^j + \sum_{i=1}^{S-1} \vy_i^j + \normal(\mathbf{0}, \sigma_v^2 \mathbb{I}_k)$.\\ 
  \If{$|\vv^j| \leq \tau$}{add $j$ to $J^*$\tcp*{$J^*$ collects $j$ that pass the norm verification.}}
  }
  Send $J^*$ to Verifiers $1,\ldots, S-1$.\\
  {\em Optional:} \If{not $\textsf{Valid}(J^*)$}{Abort\tcp*{Ensure $J^*$ is large enough.}}
  $\vs_0 = \mathbf{0}$.\\
  \For{$j \in J^*$}{
   $\vs_0 = \vs_0 + \vz_0^j$.\\
  }
  \For{$i=1,\ldots, S-1$} {
  Receive $\vs_i$ from Verifier $i$.\\
  }
  Return $\sum_{i=0}^{S-1} \vs_i$.
}

\setcounter{AlgoLine}{0}
\vspace{4pt}
\hrule
\vspace{4pt}
  \Fn{Verifier-$i$ ($i \geq 1$)}{
  \KwParam{Integer $k$. Noise scale $\sigma_{v} \in \Re$.}
  Receive $V_i = \{j, \vz_i^j)\}$ from Provers. Let $J_i = \{j : (j, \vz_i^j) \in V_i\}$.\\ 
  Receive $\vW \in \Re^{k\times d}$ from Verifier-0.\\
  \For{$j \in J_i$}{
  Compute $\vy_i^j = \vW \vz_i^j + \normal(\mathbf{0}, \sigma_v^2 \mathbb{I}_k)$.\\
  }
  Send $\{(j,\vy_i^j)\}$ to Verifier-0.\\
  Receive $J^*$ from Verifier-0.\\
  {\em Optional:} \If{not $\textsf{Valid}(J^*)$}{Abort\tcp*{Ensure $J^*$ is large enough.}}
  $\vs_i = \mathbf{0}$.\\
  \For{$j \in J^*$}{
   $\vs_i = \vs_i + \vz_i^j$.\\
  }
  Send $\vs_i$ to Verifier-0.\\
}

\end{algorithm}

The privacy proof is nearly identical to the last section. Indeed up to the computation of $J^*$, the protocol is exactly equivalent to the norm verification protocol. Verifiers other than verifier 0 do not receive any additional message after $J^*$, so that a simulator for a subset of verifiers excluding verifier 0 is essentially identical to that in the previous section. Verifier $0$ receives a set of vectors $\{s_i\}$. For $i \neq T$, the simulator simulates $s_i \sim \normal(\mathbf{0}, |J^*|\sigma_{SS}^2\mathbb{I}_d)$ subject to the sum of all $s_i$'s being equal to the output. It can be easily verified that this part of the simulation is exact. Privacy follows.

We next prove the correctness. We wish to prove that when all the parties are honest, then the sum is correctly computed except with a small failure probability. With probability $1-n\beta$, each of the $n$ norm verification steps succeed, so that $J^*$ is the set of all clients. Conditioned on this, the correctness of the secret sharing and the commutativity of addition immediately imply that the sum computed by Verifier 0 is the desired sum of all vectors.

We note that for many applications such as gradient accumulation, a weaker correctness notion may suffice. If $J^*$ is a random subset of $[n]$ with each $j$ landing in $J^*$ with probability $(1-\beta)$, we get an unbiased estimate of the sum. For this weaker definition of correctness, the failure probability does not need to be scaled by a multiplicative factor of $n$ which translates to a smaller threshold $\tau$, and thus better robustness.

Finally we argue robustness. Consider a client $j$. If the client secret-shares a vector with norm at most $\rho$, then their affect on the computed sum is clearly at most $\rho$. On the other hand, if client $j$'s shares add up to a vector of norm larger than $\rho$, it will be rejected by the norm verification step except with probability $\beta$. This means that $j \not\in J^*$ and $j$'s secret shares do not contribute at all to the compute sum. Additionally, if $j$ does not send messages to all the verifiers, their input gets rejected as well.

When the validity check on $J^*$ is added, the robustness claim is weaker. Indeed suppose that the validity check compares $|J^*|$ to a threshold, say $\frac n 2$. Then the $(\frac{n}{2}+1)$th malicious client can cause the computation to abort. The robustness guarantee now says that if the computation succeeds, then the effect of any potentially malicious client is bounded. Further, we can argue that a small number of malicious clients cannot cause the computation to abort, except with small probability.

We have thus established correctness, robustness and privacy of our protocol. For $n$ clients sending vectors in $\Re^d$, the communication cost for each client is $O(d|S|)$. The communication cost between servers is $O(dk+nk+d|S|)$. Recall that a $k= O(\sqrt{\ln n})$ suffices to get polynomially small completeness and soundness.

\subsection{On the Privacy of the Sum}
We established the privacy of the protocol, conditioned on the sum. How do we ensure the privacy of the sum itself? One option is to add differential privacy noise to the sum itself to ensure privacy. If each verifier adds noise to $s_i$, we get a differential privacy guarantee against any strict subset of the verifiers. The eventual noise variance for the sum then scales with the number of servers.

An appealing alternative is to distribute the noise generation itself. This approach goes back to~\citet{ODOpaper}. The question of generating noise on different clients such that the sum has a certain distribution has been studied for this reason. While Gaussians noise has the nice property that sum of gaussians is a gaussian, Laplace noise is also ``divisible''~\citep{Goryczka2017, Balle2020}. These arguments however require that the summation be done over real numbers. In particular, this means that for privacy to hold, the constituents of the sum may need to be communicated to sufficiently high precision even if the original vectors are $\{0,1\}$. Works such as~\citep{Agarwal18} address this question of preserving privacy while reducing the communication.

Recent results on privacy amplification by shuffling offer an elegant way out of this cononudrum. The general results in this direction~\citep{ErlingssonFMRTT19,BalleBGN19a} say that local randomizers, when shuffled give strong central differential privacy guarantees. In particular, since summation is a post-processing of shuffling, these results apply to the sum.  The privacy-accuracy trade-offs of the shuffle model are very competitive with the central model for many settings~\citep{ErlingssonFMRSTT20,FeldmanMT21}.

This ability to post-process without hurting privacy offers additional benefits. The secret-shares themselves can be rounded, truncated, or compressed without hurting privacy. For example, when the input vectors are $\{0, 1\}$, the secret sharing algorithm can use discrete gaussian noise~\citep{canonne2020discrete}, and truncate all secret shares to $[-B, B]$ for a suitable constant $B$. This does not affect the privacy claim, and the truncation operator is the identity except with a small probability depending on $B$. The small loss in accuracy due to rare truncation can be analytically or empirically traded-off against the communication cost. As an example $B=127$ would suffice for encoding each bit as 8 bits, and would ensure that the likelihood of any single bit being distorted, say for $\sigma_{SS} = 20$ is at most $10^{-8}$. This may be an acceptable error rate in applications where randomized response is used to generate the bit vectors. In comparison the field size in PRIO must grow with the number of clients and for typical values, one would use at least 32 bits.

\bibliographystyle{plainnat}
\bibliography{dzk-refs}

\end{document}